\documentclass[aps,onecolumn,a4paper,superscriptaddress,longbibliography,nofootinbib,notitlepage]{revtex4-1}

\usepackage{lipsum}
\usepackage{setspace}

\usepackage{graphicx}
\usepackage{amsmath}
\usepackage{amssymb}
\usepackage{amsthm}
\usepackage[utf8]{inputenc}
\usepackage{comment}

\usepackage{color}
\usepackage[dvipsnames]{xcolor}

\newcommand{\norm}[1]{\left\lVert#1\right\rVert} % norm: double vertical bars

\usepackage[a4paper,total={6.5in,9in}]{geometry}
\usepackage[english]{babel}
\usepackage[T1]{fontenc}
\usepackage{mathrsfs}
\usepackage{braket}
\usepackage{physics}
\usepackage{enumerate}
\usepackage{graphicx}
\usepackage{mathtools}
\usepackage[export]{adjustbox}
\usepackage{microtype}

\theoremstyle{plain}

\usepackage[colorlinks = true,
            linkcolor = black,
            urlcolor  = blue,
            citecolor = blue,
            anchorcolor = black]{hyperref}

\newtheorem{theorem}{Theorem}
\newtheorem{corollary}[theorem]{Corollary}
\newtheorem{proposition}[theorem]{Proposition}
\newtheorem{lemma}[theorem]{Lemma}
 
\newtheorem{definition}[theorem]{Definition}  
  
\newtheorem*{remark*}{Remark}   %  \newtheorem*{remark}{Remark}  * is for numbering
\renewcommand\qedsymbol{$\blacksquare$}
\newenvironment{proof-of}[1][{\hspace{-\blank}}]{{\medskip\noindent\textit{Proof~{#1}.\ }}}{\hfill\qedsymbol}

\renewcommand{\Tr}{{\operatorname{Tr}\,}}

\renewcommand{\dim}{{\operatorname{dim}}}
\newcommand{\id}{{\operatorname{id}}}

\newcommand{\1}{\openone}

\newcommand{\proj}[1]{|#1\rangle\!\langle #1|}
\newcommand{\cD}{{\mathcal{D}}}

\newcommand{\Sym}{{\textnormal{Sym}}}

\newcommand{\CB}{{\textnormal{CB}}}

\newcommand{\nc}{\newcommand}
\nc{\rnc}{\renewcommand}

\nc{\avg}[1]{\langle#1\rangle}
\nc{\Rank}{\operatorname{Rank}}
\nc{\smfrac}[2]{\mbox{$\frac{#1}{#2}$}}

\nc{\ox}{\otimes}
\nc{\dg}{\dagger}
\nc{\da}{\downarrow}
\nc{\ua}{\uparrow}
\nc{\cA}{{\cal A}}
\nc{\cB}{{\cal B}}
\nc{\cC}{{\cal C}}
\nc{\cF}{{\cal F}}
\nc{\cG}{{\cal G}}
\nc{\cH}{{\cal H}}
\nc{\cI}{{\cal I}}
\nc{\cJ}{{\cal J}}
\nc{\cK}{{\cal K}}
\nc{\cL}{{\cal L}}
\nc{\cM}{{\cal M}}
\nc{\cN}{{\cal N}}
\nc{\cO}{{\cal O}}
\nc{\cP}{{\cal P}}
\nc{\cQ}{{\cal Q}}
\nc{\cR}{{\cal R}}
\nc{\cS}{{\cal S}}
\nc{\cX}{{\cal X}}
\nc{\cY}{{\cal Y}}
\nc{\cW}{{\cal W}}
\nc{\cZ}{{\cal Z}}
\nc{\SWAP}{{\mathrm{SWAP}}}
\nc{\csupp}{{\operatorname{csupp}}}
\nc{\qsupp}{{\operatorname{qsupp}}}
\nc{\RR}{{{\mathbb R}}}
\nc{\CC}{{{\mathbb C}}}
\nc{\FF}{{{\mathbb F}}}
\nc{\NN}{{{\mathbb N}}}
\nc{\ZZ}{{{\mathbb Z}}}
\nc{\PP}{{{\mathbb P}}}
\nc{\QQ}{{{\mathbb Q}}}
\nc{\UU}{{{\mathbb U}}}
\nc{\EE}{{{\mathbb E}}}
\nc{\sdH}{{\widetilde{H}^{\downarrow}}}
\nc{\suH}{{\widetilde{H}^{\uparrow}}}
\nc{\ChN}{\Gamma_{\cN^{\otimes n}}^{{A'}^nB^n}}
\nc{\ChM}{\Gamma_{\cM}^{A'B}}
\nc{\wtD}{{\widetilde{D}}}
\nc{\wtQ}{{\widetilde{Q}}}
\nc{\pdH}{{\bar{H}^{\downarrow}}}
\nc{\puH}{{\bar{H}^{\uparrow}}}
\nc{\cAd}{\textnormal{Ad}}

\newcommand{\mids}{\;\|\;}

\begin{document}
\raggedbottom
\clubpenalty=10000
\widowpenalty=10000
%\title{Upper Bounds on the Quantum Capacity of Symmetric Codes}
\title{From Permutation Symmetry to Communication Bounds and Additivity}

\author{Zahra Baghali Khanian}
\affiliation{Perimeter Institute for Theoretical Physics, Ontario, Canada, N2L 2Y5}
\affiliation{Institute for Quantum Computing, University of Waterloo, Ontario, Canada, N2L 3G1}

\author{Debbie Leung}
\affiliation{Department of Combinatorics and Optimization, University of Waterloo, Ontario, Canada, N2L 3G1}
\affiliation{Institute for Quantum Computing, University of Waterloo, Ontario, Canada, N2L 3G1}
\affiliation{Perimeter Institute for Theoretical Physics, Ontario, Canada, N2L 2Y5}
\author{Graeme Smith}
\affiliation{Department of Applied Mathematics, University of Waterloo, Ontario, Canada, N2L 3G1}
\affiliation{Institute for Quantum Computing, University of Waterloo, Ontario, Canada, N2L 3G1}

\begin{abstract}
Correlations across channel uses can improve quantum communication rates,
making optimization over arbitrarily large blocks a central difficulty
in determining quantum capacity. We show how full permutation invariance
limits this advantage for every finite-dimensional memoryless channel.
The optimized coherent information per use converges to the single-use
maximum and consequently attains its supremum at a finite block length.
Thus finite-block superadditivity remains possible, while the asymptotic
optimization reduces to a single channel use. We also establish a strong
converse for pure-source entanglement generation with
permutation-invariant inputs: at any fixed rate above the single-use
coherent-information maximum, the fidelity tends to zero.
Beyond these communication bounds, we identify a common structure
underlying completely bounded norms and the sandwiched R\'enyi channel
quantities governing discrimination against a replacer and
entanglement-assisted communication. De Finetti reduction and permutation
covariance yield systematic alternative proofs of their known
multiplicativity and additivity results within a shared variational
framework. We further show that, for permutation-invariant inputs,
correlations cannot sustain an asymptotic reduction in output entropy
per use below the single-use minimum. This entropy limit holds for
R\'enyi orders greater than one and for the von Neumann entropy, extending
the role of symmetry from communication bounds to the control of output
entropy.

\end{abstract}
\maketitle

\section{Introduction}

How much can correlations between channel uses improve quantum
communication, and when can this improvement be controlled by an
optimization over a single use? These questions lie at the heart of the
quantum-capacity problem. For a finite-dimensional memoryless channel
$\cN$, the Lloyd--Shor--Devetak (LSD) quantum coding theorem gives
\[
 \cQ(\cN)=\lim_{n\to\infty}\frac1n Q^{(1)}(\cN^{\ox n}),
 \qquad
 Q^{(1)}(\cN)=\max_{\rho^A}I_c(\cN,\rho^A),
\]
where $I_c$ is the coherent information
\cite{Lloyd_capacity_97,Shor_direct_capacity2002,Devetak-capacity-2005}.
The regularization reflects the possibility that a correlated input
outperforms every independent-input strategy. This difficulty extends
even to detecting positive capacity: for every prescribed number of
channel uses, there are channels whose coherent information remains zero
at that block length although their quantum capacity is positive
\cite{CubittEtAl2015}. This motivates asking whether input restrictions can reduce the
optimization to a single channel use or, more generally, remove the need
to consider arbitrarily large blocks.

Repetition codes provide a particularly revealing motivation. Shor and
Smolin showed that degeneracy, i.e. the possibility that distinct physical
errors have the same action on the code space, can improve quantum
communication rates \cite{ShorSmolin1996}. DiVincenzo, Shor, and Smolin
then analyzed concatenations of cat codes with random outer codes and
connected their achievable rates to the coherent information of the
inner code states \cite{DiVincenzoShorSmolin1998}. Smith and Smolin
extended the usefulness of degenerate constructions to broad families
of Pauli channels \cite{SmithSmolin2007}. These works showed that
structured block codes can achieve rates above the single-letter
achievable rate $Q^{(1)}(\cN)$, motivating the optimization of such codes
to obtain tighter lower bounds on quantum capacity.

Other works in this vein have made the role of symmetry increasingly
explicit. Leditzky, Leung, and Smith exhibited superadditivity for the
dephrasure channel using weighted repetition inputs
\cite{LeditzkyLeungSmith2018}. These inputs have the form
$\lambda(|0\rangle\langle0|)^{\ox n}
 +(1-\lambda)(|1\rangle\langle1|)^{\ox n}$ and are invariant under
all permutations of the channel uses. Bausch and Leditzky exploited
code symmetries to study repetition, concatenated repetition, and graph
codes for arbitrary Pauli noise \cite{BauschLeditzky2021}.
More recently, Bhalerao and Leditzky used permutation symmetry to evaluate
coherent information for large blocks and improve communication-rate
and noise-threshold lower bounds for several channel families, including
through repetition constructions with nonorthogonal code states
\cite{BhaleraoLeditzky2025}. Symmetry thus makes useful families of
correlated inputs accessible to analysis and computation, while retaining
the possibility of finite-block superadditivity.

These constructive results raise a complementary question: how much
advantage can persist when the block length grows and the input remains
invariant under every permutation? We address this question uniformly
over all permutation-invariant density operators, beyond any particular
repetition-code ansatz. Here the symmetry is imposed on the input to
all $n$ channel uses. A longer code built from symmetric blocks need not remain fully
permutation invariant
\cite{DiVincenzoShorSmolin1998,BhaleraoLeditzky2025}.
Our results apply only when the input to the entire code has this symmetry.

Our first result quantifies this limitation. Define
\begin{equation}
 q_n^{\mathrm{sym}}(\cN):=
 \frac1n\max_{\rho^{A^n}}I_c(\cN^{\ox n},\rho^{A^n}),
 \label{eq:alternative-intro-symmetric-coherent}
\end{equation}
where the maximum is over all permutation-invariant input states
$\rho^{A^n}$. For every fixed finite-dimensional channel, we prove
\begin{equation}
 \begin{gathered}
 Q^{(1)}(\cN)\leq q_n^{\mathrm{sym}}(\cN)
 \leq Q^{(1)}(\cN)+O\!\left(\frac{\log n}{\sqrt n}\right),\\
 \lim_{n\to\infty}q_n^{\mathrm{sym}}(\cN)=Q^{(1)}(\cN).
 \end{gathered}
 \label{eq:alternative-intro-coherent-bound}
\end{equation}
The lower bound follows from product inputs, whereas the upper bound is
uniform over the entire symmetric input class. Thus the asymptotic
coherent-information optimization under full permutation invariance
requires only a single channel use. Finite-block superadditivity remains
possible, but its excess per use vanishes at a controlled rate.

The result also gives a precise finite-block consequence. For any fixed
$\delta>0$, only finitely many block lengths can yield
$q_n^{\mathrm{sym}}(\cN)\geq Q^{(1)}(\cN)+\delta$.
Moreover, the largest normalized coherent information within this
symmetric family is attained at some finite block length. Indeed, if one
block gives a strict improvement over $Q^{(1)}(\cN)$, convergence places
all sufficiently long blocks below that value; otherwise a single use
already attains the supremum. This provides an asymptotic benchmark for
searches over symmetric codes and shows that an arbitrarily late gain
of fixed size is excluded for a fixed channel. 
%The maximizing length can depend on the channel, and the asymptotic estimate alone does not supply an explicit stopping rule for finding it.

We show that pure-source entanglement-generation codes with fully
permutation-invariant channel inputs obey a strong-converse bound:
at any fixed communication rate above $Q^{(1)}(\cN)$, their fidelity tends
to zero as the block length grows. Thus $Q^{(1)}(\cN)$ is a single-letter
upper bound on the quantum capacity of this code class. The precise
code definition is given in Section~\ref{sec:symmetric-coding}; the
result also applies to entanglement transmission with an isometric
encoder and a permutation-invariant channel input.
Our bound holds for every finite-dimensional channel. 
%The symmetry assumption is imposed on the input, whereas the pretty strong converse of Ref.~\cite{pretty_strong_C_Q} assumes a degradable channel. Showing that the capacity of our restricted code class equals $Q^{(1)}(\cN)$ would additionally require an achievability proof within that class.

The proof develops a route from symmetry to tensor-power bounds using
quantum de Finetti and postselection techniques
\cite{half_de_Finetti_2007,post-selection2009}.
A symmetric purification of the input is dominated by a polynomial
multiple of a mixture of tensor-power pure states. Retaining a classical
label for this mixture allows a sandwiched R\'enyi divergence estimate
to reduce the problem to single-copy inputs. The logarithm of the
postselection factor, combined with a uniform estimate as the R\'enyi
order approaches one, yields the stated finite-length correction.
The resulting proof tools are useful beyond the coherent-information
problem, because many channel quantities face the same difficulty of
optimizing over correlations across multiple uses.

The second part of the paper makes this connection systematic.
Completely bounded norms, channel discrimination against a replacer,
and entanglement-assisted communication involve quantities with closely
related variational formulas. Their tensor-product behavior matters
operationally: the additivity of sandwiched R\'enyi channel quantities
enters strong-converse bounds for entanglement-assisted communication
\cite{GuptaWilde2015} and the analysis of discrimination between a channel
and a replacer that discards its input and prepares a fixed output
\cite{CooneyMosonyiWilde2016}. The corresponding multiplicativity and
additivity results are already known
\cite{DevetakJungeKingRuskai2006,Jencova2006,GuptaWilde2015,CooneyMosonyiWilde2016};
here we identify a common theme by deriving them through de Finetti
reduction and permutation covariance.
A common derivation of their additivity
therefore clarifies which structural properties control the effect of
correlated inputs across these tasks.
Minimum output R\'enyi entropy fits into this framework through the
pure-input restriction of the variational formula for the completely
bounded norm, followed by a logarithmic rescaling. The same postselection
method bounds how much correlations in permutation-invariant inputs can
reduce the output entropy.

\medskip

The organization of the paper is as follows.
We introduce our notation and preliminaries at the end of this section.
Section~\ref{sec:coherent-bound} establishes the coherent-information
bound for permutation-invariant inputs. Section~\ref{sec:symmetric-coding}
derives the strong converse for the corresponding class of codes.
Section~\ref{sec:applications} introduces the common functional and
applies the postselection method to additivity and minimum output entropy.
Section~\ref{sec:discussion} discusses the implications of these results
and their common proof strategy. Appendix~\ref{app:supporting-inequalities}
collects the supporting inequalities.

\medskip
\emph{Notation and preliminaries.} All Hilbert spaces are finite dimensional and all logarithms are to base two,
except when $\ln$ is written explicitly. We write $\mathcal S(A)$ for density
operators on $A$, $d=\dim A$, and $\|X\|_p=(\Tr|X|^p)^{1/p}$.
The von Neumann entropy and Umegaki relative entropy are
$S(\rho)=-\Tr\rho\log\rho$ and
$D(\rho\mids\tau)=\Tr\rho(\log\rho-\log\tau)$, respectively;
$S(R|B)_\omega=S(\omega^{RB})-S(\omega^B)$.
For a state $\rho$ and a positive comparison operator $\tau$, the sandwiched
R\'enyi divergence of order $\alpha>1$ is
\begin{align}
 \widetilde D_\alpha(\rho\mids\tau)
 =\frac{1}{\alpha-1}\log\Tr
 \left(\tau^{\frac{1-\alpha}{2\alpha}}
       \rho\tau^{\frac{1-\alpha}{2\alpha}}\right)^\alpha.
 \label{eq:sandwiched-definition}
\end{align}
It is infinite unless $\operatorname{supp}\rho\subseteq\operatorname{supp}\tau$;
negative powers are evaluated on the support. The comparison operator need
not have trace one. The conditional entropies used below are
\begin{align}
 \sdH_\alpha(R|B)_\omega
 &=-\widetilde D_\alpha(\omega^{RB}\mids\1^R\ox\omega^B),\nonumber\\
 \suH_\alpha(R|B)_\omega
 &=-\inf_{\tau_B\in\mathcal S(B)}
       \widetilde D_\alpha(\omega^{RB}\mids\1^R\ox\tau_B).
\end{align}
We use data processing and monotonicity in the R\'enyi order
\cite{Beigi2013,Tomamichel_book}. In particular,
$-S(R|B)_\omega\leq-\suH_\alpha(R|B)_\omega$ for $\alpha>1$.
For a pure $\psi^{RA}$ with $R\cong A$, put
\begin{align}
 Q_\alpha^{\downarrow}(\cN)
 :=\max_{\psi^{RA}\ {\rm pure}}
 \widetilde D_\alpha((\id_R\ox\cN)(\psi^{RA})
                   \mids\1^R\ox\cN(\psi^A)).
 \label{eq:Qalpha-definition}
\end{align}
The arrow specifies the fixed output marginal in the comparison operator;
this quantity is distinct from an optimization over that marginal.

\section{Bound on the coherent information of permutation invariant states}
\label{sec:coherent-bound}

It is known that if a state $\rho^{A^n}$ is permutation-invariant, then it has a purification $\ket{\rho}^{R^nA^n}$ whose support lies in the symmetric subspace
$\Sym^n (\cH_A \ox \cH_R)$ where $\cH_A \cong \cH_R$ \cite{half_de_Finetti_2007}.
This pure state $\ket{\rho}^{R^n A^n}$ is permutation-invariant as
\begin{align}
\pi^{(RA)^n}  \ket{\rho}^{R^n A^n}=(\pi^{R^n} \otimes \pi^{A^n} ) \ket{\rho}^{R^n A^n} = \ket{\rho}^{R^n A^n}
\quad \text{for all } \pi \in S_n,
\end{align}
where $S_n$ denotes the symmetric group. 
%Let $\rho$ be a pure state  supported on the symmetric subspace $\text{Sym}^n(\mathcal{H}_R \ox \mathcal{H}_A)$, 
By  post-selection lemma (Lemma~\ref{lemma: post-selection}) the state is  bounded as
\begin{align}\label{eq: de Finetti for rho}
\proj{{\rho}}^{R^n A^n }
   &\leq  n^{|\cH_A|^2} \int  (\proj{\varphi}^{RA})^{\ox n} d\varphi,
\end{align}
where $d\varphi$ is the uniform measure on $\Sym^n (\cH_A \ox \cH_R)$ and the right-hand-side of the inequality is proportional to the projection onto the  symmetric subspace
$\Sym^n (\cH_A \ox \cH_R)$.

Here the measure $d\varphi$ is understood as Haar probability measure on
single-copy pure states of $R\otimes A$; their tensor powers lie in the
symmetric subspace. With $d=\dim\cH_A$, the exact post-selection factor is
\[
 g_{n,d}=\binom{n+d^2-1}{d^2-1},\qquad
 \int(\proj{\varphi}^{RA})^{\ox n}d\varphi
 =\frac{P_{\Sym^n(RA)}}{g_{n,d}},
\]
as stated in Lemma~\ref{lemma: post-selection}. The polynomial estimates
in this section use $n\geq2$, for which $g_{n,d}\leq n^{d^2}$ and
Carath\'{e}odory's theorem gives $L\leq g_{n,d}^2\leq n^{2d^2}$.
%At $n=1$, use $g_{1,d}=d^2$ instead of $n^{d^2}$; this endpoint does not affect the asymptotic conclusions.

% 
By Carath\'{e}odory’s Theorem, w.l.o.g., the state appearing on the right-hand-side of (\ref{eq: de Finetti for rho}) can be expressed as 
supported on $L \leq  n^{2 |\cH_A|^2} $ points
\begin{align} 
\int  (\proj{\varphi}^{RA})^{\ox n} d\varphi=  \sum_{l=1}^L q_l \; (\proj{\varphi_l}^{RA})^{\ox n}.
\end{align}
We apply the inequality in Eq.~(\ref{eq: de Finetti for rho}) to obtain an upper bound on
the $\alpha$-R\'{e}nyi generalization of the coherent information. 

%\begin{align}
%\overline{\varphi}^{R^n A^n}:=\int  \proj{\varphi}^{\ox n} d\varphi
 %  &=  \sum_{l=1}^L q_l \; (\proj{\varphi_l}^{RA})^{\ox n}, 
%\end{align}

%Define
%\begin{align}
%\overline{\varphi}^{XR^n A^n}:= \sum_{l=1}^L q_l\; \proj{x}^X\ox \proj{\varphi_l}^{\ox n}, 
%\end{align}

\begin{lemma}\label{lemma: H(rho) < H(varphi_l)}
Let $\ket{\rho}^{R^n A^n}$ be a pure quantum state that is permutation-invariant, i.e.,
\[
(\pi^{R^n} \otimes \pi^{A^n} ) \ket{\rho}^{R^n A^n} = \ket{\rho}^{R^n A^n}
\quad \text{for all } \pi \in S_n.
\]
For $\alpha >1$, the following holds for channel $\cN$ and state ${\sigma}^{R^nB^n}=(\id_{R^n}\ox \cN^{\ox n}) (\proj{\rho}^{R^nA^n})$: 
\begin{align}
-\suH_{\alpha}(R^n|B^n)_{\sigma}
\leq n\; \max_{\ket{\psi}^{AR}} \widetilde{D}_{\alpha}\left( (\id_{R} \ox \cN)(\proj{\psi}^{RA}) \mids \1^{R} \ox \cN(\psi^{A}) \right)+\frac{\alpha}{\alpha-1} \cdot |\cH_A|^2\log n.
 \end{align}
\end{lemma}

\begin{proof}

The state $\rho^{R^nA^n}$ is related to state $\overline{\varphi}^{R^nA^n}$ as in Eq.~(\ref{eq: de Finetti for rho})
\begin{align*}
\proj{\rho}^{R^nA^n}  \leq n^{|\cH_A|^2} \underbrace{\sum_{l=1}^L q_l \; \proj{\varphi_l}^{\ox n}}_{\overline{\varphi}},
\end{align*}
and $L \leq  \; n^{2|\cH_A|^2} $.
By the positivity of the channel, applying it preserves the inequality and we obtain
\begin{align}
\sigma^{R^nB^n}&=(\id_{R^n} \ox \cN^{\ox n}) \left(\proj{\rho}^{R^nA^n} \right) \\
& \leq n^{ |\cH_A|^2} \sum_{l=1}^L q_l \;  \left( (\id_{R} \ox \cN) (\proj{\varphi_l}) \right)^{\ox n}\\
& = n^{|\cH_A|^2} \sum_{l=1}^L q_l \; (\upsilon_l^{RB})^{\ox n} \\
& = n^{|\cH_A|^2} \overline{\upsilon}^{R^nB^n}.
\end{align}
We apply Lemma~\ref{lemma: H_alpha for rho<gamma rho'} for the inequality $\sigma^{R^nB^n} \leq n^{ |\cH_A|^2} \; \overline{\upsilon}^{R^nB^n}$
\begin{align}
-\suH_{\alpha}(R^n|B^n)_{\sigma}
&\leq \widetilde{D}_{\alpha}( \overline{\upsilon}^{R^nB^n} \mids \1^{R^n} \ox \overline{\upsilon}^{B^n}) +\frac{\alpha}{\alpha-1} |\cH_A|^2 \log n \\
&\leq \widetilde{D}_{\alpha}( \overline{\upsilon}^{R^nB^nY} \mids \1^{R^n} \ox \overline{\upsilon}^{B^nY}) +\frac{\alpha}{\alpha-1} |\cH_A|^2 \log n\\
&\leq \max_l \; \widetilde{D}_{\alpha}( (\upsilon_l^{RB})^{\ox n} \mids \1^{R^n} \ox (\upsilon_l^{B})^{\ox n})+\frac{\alpha}{\alpha-1} |\cH_A|^2 \log n\\
&= \max_l \; n\; \widetilde{D}_{\alpha}( \upsilon_l^{RB} \mids \1^{R} \ox \upsilon_l^{B})+\frac{\alpha}{\alpha-1} |\cH_A|^2 \log n,
 \end{align}
%where in the first line we apply Lemma~\ref{lemma: H_alpha for rho<gamma rho'}.
where the second line is due to data processing and the entropy is with respect to the state
\begin{align*}
\overline{\upsilon}^{R^n B^nY} = \sum_{l=1}^L q_l \; (\upsilon_l^{RB})^{\ox n} \ox \proj{l}^Y,
\end{align*}
and the last inequality follows from
Lemma~\ref{lemma: D < max_y D_y}.
%
%There are $L \leq  n^{2|\cH_A|^2} $ terms in the average state $\varphi$. 
%
The state above is defined as $\upsilon_l^{RB}=(\id_{R} \ox \cN)\proj{\varphi_l}^{RA}$, hence
\begin{align}
 \widetilde{D}_{\alpha}( \upsilon_l^{RB} \mids \1^{R} \ox \upsilon_l^{B}) &\leq \max_{\ket{\psi}^{AR}} \widetilde{D}_{\alpha}\left( (\id_{R} \ox \cN)\psi^{RA} \mids \1^{R} \ox \cN(\psi^{A}) \right).
 %&= \max_{\ket{\psi}^{AR}} I_c(\cN, \psi).
\end{align}
\end{proof}

\begin{theorem}\label{thm: bound on permutation invariant states}
Let $\ket{\rho}^{R^n A^n}$ be a pure quantum state that is permutation-invariant, i.e.,
\[
(\pi^{R^n} \otimes \pi^{A^n} ) \ket{\rho}^{R^n A^n} = \ket{\rho}^{R^n A^n}
\quad \text{for all } \pi \in S_n.
\]
 Then, for a channel $\cN$ and state ${\sigma}^{R^nB^n}=(\id_{R^n}\ox \cN^{\ox n}) \rho^{R^nA^n}$
 the following holds for $\alpha = 1 + \frac{1}{\sqrt{n}}$ and any $n \in \mathbb{N}$
\begin{align*}
-\suH_{\alpha}(R^n|B^n)_{\sigma}
\leq n \; Q^{(1)}(\cN)+O(\sqrt{n} \log n)
 \end{align*}

\end{theorem}
\begin{proof}
In Lemma~\ref{lemma: H(rho) < H(varphi_l)}, we obtain the following bound for  $\alpha> 1$
\begin{align}\label{eq: suH_alpha < n D_alpha}
-\suH_{\alpha}(R^n|B^n)_{\sigma}
\leq n\; \max_{\ket{\psi}^{AR}} \widetilde{D}_{\alpha}\left( (\id_{R} \ox \cN)\psi^{RA} \mids \1^{R} \ox \cN(\psi^{A}) \right)+\frac{\alpha}{\alpha-1} \cdot |\cH_A|^2\log n.
 \end{align}
%
%The smooth max-divergence is bounded using the Taylor series expansion for the R\'enyi divergence. 
It is known that for all $\alpha \in (1,2]$ and all $\rho$ and $\sigma$, there exists a constant $C$ such that \cite[Eq.~(4.102)]{Tomamichel_book}
\begin{align}
\widetilde{D}_\alpha(\rho \Vert \sigma) 
\leq D(\rho \Vert \sigma) 
+ (\alpha - 1) \frac{1}{2 \log(e)} V(\rho \Vert \sigma) 
+ (\alpha - 1)^2 C,
\end{align}
where $V(\rho  \Vert \sigma):=\Tr\left(\rho (\log \rho -\log \sigma)^2 \right)-(D(\rho  \Vert \sigma))^2$.
By applying this inequality to Eq.~(\ref{eq: suH_alpha < n D_alpha}) and setting $\alpha \coloneqq 1 + \frac{1}{\sqrt{n}}$ we obtain
\begin{align}
n &\;\widetilde{D}_{\alpha}\left( (\id_{R} \ox \cN)\psi^{RA} \mids \1^{R} \ox \cN(\psi^{A}) \right)\\
&\leq n\; D\left( (\id_{R} \ox \cN)\psi^{RA} \mids \1^{R} \ox \cN(\psi^{A}) \right)+  \frac{\sqrt{n}}{2 \log(e)} V\left( (\id_{R} \ox \cN)\psi^{RA} \mids \1^{R} \ox \cN(\psi^{A}) \right) + C \\
&\leq n\; \max_{\psi^{A}} I_c(\cN, \psi)+  \frac{\sqrt{n}}{2 \log(e)} V\left( (\id_{R} \ox \cN)\psi^{RA} \mids \1^{R} \ox \cN(\psi^{A}) \right) + C.
\end{align}

Put $\omega_\psi=(\id_R\ox\cN)(\proj\psi^{RA})$,
$s_\psi=\1^R\ox\omega_\psi^B$, and $d=\dim R=\dim A$.
The relations $\omega_\psi\leq ds_\psi$ and $\Tr s_\psi=d$ imply
$\mathbb E\mathrm e^{\pm Z_\psi}\leq d$ for the associated
Nussbaum--Szko{\l}a log-likelihood variable $Z_\psi$, with natural
logarithms~\cite[Sec.~4.4.2]{Tomamichel_book}.
These moment bounds uniformly control its variance and the third
derivative of $\ln\mathbb E\mathrm e^{tZ_\psi}$ for
$0\leq t\leq\tfrac12$. Taylor expansion of the Petz divergence,
which upper bounds the sandwiched divergence, therefore gives
$V(\omega_\psi\mids s_\psi)\leq V_*<\infty$ and permits one finite
remainder constant $C_*$ for all inputs and $0<t\leq\tfrac12$
\cite[Eqs.~(4.88) and (4.102)]{Tomamichel_book}.
Maximizing the preceding bound thus adds at most
$\sqrt n V_*/(2\log(e))+C_*=O(\sqrt n)$ for $n\geq4$;
Eq.~\eqref{eq: suH_alpha < n D_alpha} then proves the claim.

\end{proof}
We note that for any $\alpha>1$, any finite dimensional state $\sigma$, inequality $-S(R^n|B^n)_{\sigma} \leq -\suH_{\alpha}(R^n|B^n)_{\sigma}$ holds \cite{Tomamichel_book}. 
Hence we can conclude the following. 
\begin{corollary} %\label{thm: bound on permutation invariant states}
For any  channel $\cN$  and any number $n \in \mathbb{N}$
\begin{align*}
\sup_{\rho^{A^n}} \; I_c(\cN^{\ox n} , \rho^{A^n})
\leq n \; Q^{(1)}(\cN)+O(\sqrt{n} \log n)
 \end{align*}
where the optimization is over all permutation-invariant states $\rho^n$. 
\end{corollary}

%%%%%%%%%%%%%%%%%%%%%%%%%%%%%%%%%%%%%%%%%%%%%%%%%

\section{Strong converse for the quantum capacity with symmetric codes}
\label{sec:symmetric-coding}

Obtaining a strong converse bound for the entanglement generation capacity 
implies that such a bound also holds for the entanglement transmission capacity. 
We begin by briefly defining the entanglement generation code below.

The symmetric-input converse below concerns the pure-source codes defined
here, including entanglement transmission with an isometric encoder and a
permutation-invariant channel input. Under this input restriction, the
``w.l.o.g.'' pure-state reduction should not be taken to cover arbitrary
mixed encoders: decomposing a mixed source into pure states need not
preserve the symmetry of its input marginal.
For unrestricted inputs, a pure component with no larger error can be
selected as in~\cite[Sec.~II, Remark]{pretty_strong_C_Q}, but this selection
need not preserve the input-symmetry constraint imposed here.

\noindent \textbf{The entanglement generation code: }
An $(n,\epsilon)$ entanglement generation code is defined as follows. Alice prepares w.l.o.g a bipartite pure state  $\ket{\psi}^{A^nR}$, and sends $A^n$ to Bob through
$n$ uses of the channel $\cN: A \to B$ with an isometric extension $U_{\cN}: A \hookrightarrow BE$. The output state of the isometric extension $U_{\cN}^{\ox n}$ of the channel is
\begin{align}
    \ket{\sigma}^{B^nE^nR}:=(U_{\cN}^{\ox n}\ox \1_{R})\ket{\psi}^{A^nR},
\end{align}
where $E^n$ is the environment system of the channel.
Then, Bob applies the decoding map $\cD_n: B^n \to \widetilde{A}$. The output state of isometric extension $U_{\cD_n}: B^n \hookrightarrow \widetilde{A} W$ of the decoding map is 
\begin{align}
    \ket{\xi}^{\widetilde{A}E^nWR}:=(U_{\cD_n} \ox \1_{E^nR})\ket{\sigma}^{B^nE^nR},
\end{align}
where $W$ is the environment system of the decoding isometry. 
The code has \textit{error} $\epsilon$ if
\begin{align}
    P(\proj{\Phi}^{\widetilde{A}R},\xi^{\widetilde{A}R}) \leq \epsilon, %=\sqrt{1-F^2(\proj{\Phi}^{QR},\xi^{QR})}=\epsilon,
\end{align}
%
%The fidelity  for the entanglement generation protocol is defined as
%\begin{align}
%    F(\proj{\Phi}^{QR},\xi^{QR})=\sqrt{\bra{\Phi}\xi^{QR}\ket{\Phi}^{QR}},
%\end{align}
where ${\xi}^{\widetilde{A}R}=\Tr_{E^n W} (\proj{\xi}^{\widetilde{A}E^nWR})$, and ${\Phi}^{\widetilde{A}R}$ is a maximally entangled state of dimension $|\widetilde{A}|=|R|$.
The rate of the code is defined as $\frac{\log |\widetilde{A}|}{n}$.
%
%The maximum dimension $|A|$ such that there exists a
%
Let $N(n, \epsilon, \rho^n)=|\widetilde{A}|$  denote the dimension of 
quantum code for $\cN^{\ox n}$ with error $\epsilon$ and input $\rho^n=\Tr_R(\proj{\psi}^{A^n R})$. 
The (asymptotic) quantum capacity $\cQ(\cN)$ of a quantum channel $\cN$ is given by:
\begin{align}
    \cQ(\cN) = \inf_{\varepsilon > 0} \liminf_{n \to \infty} \;\frac{1}{n} \;\sup_{\rho^n} \; \log N(n, \varepsilon,\rho^n).
\end{align}

In this capacity expression, $N(n,\epsilon,\rho^n)$ is understood as the
largest achievable dimension over preparation and decoding choices with
the specified input and error bound; equivalently, the supremum is over
admissible codes. In the theorem, $0\leq\epsilon<1$ is fixed as
$n\to\infty$. The implicit constant may depend on $\epsilon$ and the
fixed channel dimension, but not on the symmetric input state.

\begin{theorem}\label{thm:symmetric-code-converse} %\label{thm: bound on permutation invariant states}
For any channel $\cN$, the rate of an entanglement generation code with length $n \in \mathbb{N}$, error $\epsilon$, and a permutation-invariant input state ${\rho}^{A^n}$ is bounded as 
\begin{align*}
\frac{1}{n}{\log}\,N(n, \epsilon, \rho^n)
\leq  \; Q^{(1)}(\cN)+O(\frac{\log n}{\sqrt{n}}).
 \end{align*} 
\end{theorem}

\begin{proof}

Let $\sigma^{RB^nE^n}$ be the output of the channel with input $\rho^n$. 
We note that for the maximally entangled state $\Phi^{RA}$ of dimension $N$, $\widetilde{H}_{\beta}^{\uparrow}(R|A)_{\Phi}=-\log N$ for  $\beta > 0$. The preceding equality
is meant for $\beta\geq\tfrac12$, which includes the range used below;
it does not hold for all $\beta>0$ under the optimized sandwiched
R\'enyi definition. For $\tfrac12<\beta<1$, we use the same trace
formula and conditional-entropy optimization as in the preliminaries,
with $\beta$ in place of $\alpha$. For  $\beta \in ( \frac{1}{2},1)$, we apply the generalization of van Dam-Hayden inequality (in Proposition 2.8. in \cite{Leditzky2016}) to obtain
\begin{align}
   \log N(n,\epsilon,\rho^n) &= -\widetilde{H}_{\beta}^{\uparrow}(R|A)_{\Phi} \\
    & \leq -\widetilde{H}_{\alpha}^{\uparrow}(R|A)_{\xi}+\frac{2\beta}{\beta-1}\log F(\Phi^{AR},\xi^{AR})\\
    & \leq -\widetilde{H}_{\alpha}^{\uparrow}(R|B^n)_{\sigma}+\frac{2\beta}{\beta-1}\log F(\Phi^{AR},\xi^{AR}),
\end{align}
where $\alpha(\beta)=\frac{\beta}{2\beta-1}>1$, and the last line is due to data processing. 
The fidelity is related to purified distance as $F(\Phi^{AR},\xi^{AR})=\sqrt{1-P(\Phi^{AR},\xi^{AR})^{2}} \geq \sqrt{1-\epsilon^2}$. Hence, we obtain
\begin{align}
   \log N(n,\epsilon,\rho^n) 
    & \leq -\widetilde{H}_{\alpha}^{\uparrow}(R|B^n)_{\sigma}+\frac{2\alpha}{1-\alpha}\log F(\Phi^{AR},\xi^{AR})\\
    & \leq -\widetilde{H}_{\alpha}^{\uparrow}(R|B^n)_{\sigma}+\frac{\alpha}{1-\alpha}\log (1-\epsilon^2).
\end{align}

To apply Section~\ref{sec:coherent-bound}, choose the symmetric
purification of $\rho^{A^n}$ on $R_0^{\otimes n}A^n$, where $R_0\cong A$.
It and the code's purification are related by an isometry between their
reference supports. Conditional R\'enyi entropy is invariant under this
isometry and adjoining unused reference dimensions. Thus the bound from
Section~\ref{sec:coherent-bound} applies to the reference $R$ used here;
no full-space isometry $R\to R_0^{\otimes n}$ is required.

For an explicit remainder uniform over the input, put $d=\dim A$ and
$g_{n,d}=\binom{n+d^2-1}{d^2-1}$. The reduction in
Section~\ref{sec:coherent-bound}, with the exact factor from
Lemma~\ref{lemma: post-selection}, gives, for every $\alpha>1$,
\begin{align}
 -\suH_\alpha(R|B^n)_\sigma
 \leq nQ_\alpha^{\downarrow}(\cN)
       +\frac{\alpha}{\alpha-1}\log g_{n,d}.
 \label{eq:symmetric-code-renyi}
\end{align}

We let $\alpha=1+\frac{1}{\sqrt{n}}$, and 
for a permutation invariant state $\rho^{A^n}$ we apply Theorem~\ref{thm: bound on permutation invariant states} to obtain a bound 
on $-\widetilde{H}_{\alpha}^{\uparrow}(R|B^n)_{\sigma}$
\begin{align}
   \log N(n,\epsilon,\rho^{A^n}) 
    & \leq n \; Q^{(1)}(\cN)+O(\sqrt{n} \log n)-(\sqrt{n}+1)\log (1-\epsilon^2).
\end{align}
Dividing by $n$ proves the stated rate bound for fixed $\epsilon<1$;
the finitely many smaller block lengths do not affect that asymptotic
statement.

To make the strong-converse conclusion explicit, consider a sequence of
such codes of dimension $M_n$, rate $r_n=n^{-1}\log M_n$, and actual
purified-distance error $\epsilon_n$. Write
$f_n:=F(\Phi,\xi_n)^2=1-\epsilon_n^2$.
{The $O(\sqrt n\log n)$ remainder in the preceding bound
can be chosen nonnegative and independently of the input state and
$\epsilon_n$.} The same inequality gives
\begin{align}
 f_n\leq\min\left\{1,\,
 2^{-\frac{n[r_n-Q^{(1)}(\cN)]-{O(\sqrt n\log n)}}{\sqrt n+1}}\right\}.
 \label{eq:original-proof-strong-converse}
\end{align}
If $\liminf_n r_n>Q^{(1)}(\cN)$, then ${O(\sqrt n\log n)/n\to0}$, so $f_n\to0$
and $\epsilon_n\to1$. The case $f_n=0$ already satisfies this bound.
This proves a strong-converse upper threshold for the specified code
class. Identifying its capacity with $Q^{(1)}(\cN)$ would additionally
require an achievability argument within that class.

\end{proof}

%%%%%%%%%%%%%%%%%%%%%%%%%%%%%%%%%%%%%%%%%%%%%%%%%%%%%%%%%%%%%%%%%%%%%%%%%%%%%%
\section{Additivity and bounds for quantum information-processing tasks}
\label{sec:applications}

The postselection method of Section~\ref{sec:coherent-bound} provides a
way to control correlations across channel uses through single-copy
states. Here we apply it to quantities that enter strong-converse bounds
for entanglement-assisted classical communication~\cite{GuptaWilde2015}
and discrimination between a channel and a replacer~\cite{CooneyMosonyiWilde2016}.
For $1<\alpha<\infty$, our objectives are multiplicativity of the
completely bounded $1\to\alpha$ norm and additivity of the sandwiched
R\'enyi channel divergence against a replacer and of channel R\'enyi
mutual information. These properties are known
\cite{DevetakJungeKingRuskai2006,Jencova2006,GuptaWilde2015,CooneyMosonyiWilde2016};
we give alternative proofs that identify the common structural features
responsible for them.

The three quantities admit a common weighted Schatten-norm expression,
which we denote by $\omega_\alpha$ and define below. Maximizing over the
input with the identity as output weight gives the completely bounded
norm. Fixing the weight to a density operator gives the divergence
against the corresponding replacer after input maximization and a
logarithmic rescaling. Minimizing over output density operators before
maximizing over the input gives the channel R\'enyi mutual information
with the same rescaling.

The functional factors on product inputs and output weights and is
concave in its input argument. Together with the permutation covariance
of every tensor-power map, this concavity allows symmetric inputs to be
selected in the relevant optimizations without restricting their values.
Postselection then bounds the contribution of correlated inputs by a
polynomial factor times values on tensor-power inputs. The commutation
of the universal postselection marginal with these tensor-power inputs
supplies the operator estimate needed for this bound. Repeating the
maps and taking roots removes the polynomial factor, proving exact
multiplicativity for fixed output weights and hence the first two
objectives. For channel R\'enyi mutual information, mixtures of
comparison states preserve the original max--inf optimization order;
additivity of state R\'enyi mutual information supplies the reverse
inequality~\cite{Beigi2013}. The same amplification argument then proves
channel additivity.

Minimum output R\'enyi entropy is a further application of this common
representation, obtained by restricting the input of $\omega_\alpha$
to pure states and applying a logarithmic rescaling. Averaging need not
preserve this purity constraint. In
Subsection~\ref{subsec:symmetric-minimum-output}, we therefore impose
permutation invariance on the channel input and apply postselection
directly. This gives a logarithmic bound on the total entropy reduction
due to correlations for fixed $\alpha>1$ and convergence of the minimum
output entropy per use to the single-use minimum, including at
$\alpha=1$.

\subsection{A common Choi representation}
\label{subsec:common-choi}

Let $\cM:A\to B$ be completely positive. Fix a basis of $A$ and its copy
$A'$, and let
\[
 |\Gamma\rangle^{A'A}=\sum_{i=1}^d|i\rangle^{A'}|i\rangle^A,
 \qquad J_\cM=(\id_{A'}\ox\cM)(|\Gamma\rangle\langle\Gamma|).
\]
We use the unnormalized maximally entangled vector; in particular,
$J_{\cM^{\ox n}}=J_\cM^{\ox n}$ after regrouping tensor factors.

\begin{lemma}\label{lem:choi-parametrization}
Every normalized pure $|\psi\rangle^{RA}$ can be written
\[
 |\psi\rangle=(V\rho_{A'}^{1/2}\ox\1_A)|\Gamma\rangle,
 \qquad \rho_{A'}=(\psi^A)^T,
\]
where $V$ is an isometry on $\operatorname{supp}\rho_{A'}$.
\end{lemma}
\begin{proof}
To see this, express $|\psi\rangle=(X\ox\1)|\Gamma\rangle$ and
use the polar decomposition $X=V\sqrt{X^\dagger X}$.
The same argument applies to $n$ copies, allowing a general correlated
$\rho_{{A'}^n}$; it imposes no product assumption on the input.
\end{proof}

For $1<\alpha<\infty$, a density operator $\rho$ on $A'$, and a positive
 definite operator $\sigma_B$, define the functional
\begin{align}
 \omega_\alpha(\rho,\sigma,\cM)
 &:=\left\|(\rho^{1/(2\alpha)}\ox\sigma^{(1-\alpha)/(2\alpha)})
 J_\cM(\rho^{1/(2\alpha)}\ox\sigma^{(1-\alpha)/(2\alpha)})\right\|_\alpha
 \nonumber\\
 &=\left\|J_\cM^{1/2}
       (\rho^{1/\alpha}\ox\sigma^{(1-\alpha)/\alpha})
       J_\cM^{1/2}\right\|_\alpha.
 \label{eq:common-omega}
\end{align}
Here $J_\cM=\ChM$. The second equality follows because $XX^\dagger$ and
$X^\dagger X$ have the same nonzero eigenvalues; the exponent on $\sigma$
is $(1-\alpha)/\alpha$ in this second expression. Normalization of $\sigma$
is not required, so the choice $\sigma=\1_B$ is allowed.
For a channel $\cN$, the Choi parametrization and cancellation of the
reference powers give
\begin{align}
 \widetilde D_\alpha((\id_R\ox\cN)(\psi)
                          \mids\psi_R\ox\sigma_B)
 =\frac{\alpha}{\alpha-1}\log\omega_\alpha(\rho,\sigma,\cN).
 \label{eq:common-divergence}
\end{align}
%The coefficient of the logarithm is $\alpha/(\alpha-1)$.
The four quantities studied here have the following common form:
\begin{align}
 \|\cM\|_{\CB,1\to\alpha}
 &=\max_{\rho\in\cD(A')}\omega_\alpha(\rho,\1_B,\cM),
 \label{eq:omega-cb-identification}\\
 \widetilde D_\alpha(\cN\mids\cR_\sigma)
 &=\frac{\alpha}{\alpha-1}\log
       \max_{\rho\in\cD(A')}\omega_\alpha(\rho,\sigma,\cN),
 \label{eq:omega-replacer-identification}\\
 \widetilde I_\alpha(\cN)
 &=\frac{\alpha}{\alpha-1}\log
       \max_{\rho\in\cD(A')}\inf_{\substack{\sigma>0\\\Tr\sigma=1}}
                \omega_\alpha(\rho,\sigma,\cN),
 \label{eq:omega-mi-identification}\\
 h_\alpha(\cN)
 &=\frac{\alpha}{1-\alpha}\log
       \max_{\substack{\rho\in\cD(A')\\\operatorname{rank}\rho=1}}
                \omega_\alpha(\rho,\1_B,\cN).
 \label{eq:omega-minimum-output-identification}
\end{align}
In the first identity $\cM$ is completely positive; in the remaining
identities $\cN$ is trace preserving as well. The comparison operators
in the second and third identities are density operators.
The first identity is the variational CB-norm formula
\cite{DevetakJungeKingRuskai2006}; the second and third follow from
Eq.~\eqref{eq:common-divergence}
\cite{CooneyMosonyiWilde2016,GuptaWilde2015}.
For the first three quantities, CB-norm multiplicativity was proved in
\cite[Theorem~11]{DevetakJungeKingRuskai2006}, and additivity of channel
R\'enyi mutual information was proved in \cite[Lemma~6]{GuptaWilde2015}.
Additivity of the divergence against a replacer follows from the former
result and the CB-norm identity in \cite[Lemma~8]{CooneyMosonyiWilde2016}.
Jen\v{c}ov\'{a} also gave an alternative proof of CB $1\to\alpha$-norm
multiplicativity for quantum channels by relating it to the
$\alpha\to\alpha$ Schatten norm of the conjugate channel~\cite{Jencova2006}.
Here we provide alternative proofs of these three properties using
$\omega_\alpha$ and the postselection method of Section~\ref{sec:coherent-bound}.
The fourth quantity is the minimum output R\'enyi entropy,
$h_\alpha(\cN):=\min_{\tau\in\cD(A)}S_\alpha(\cN(\tau))$. % where $S_\alpha(\eta)=(1-\alpha)^{-1}\log\Tr\eta^\alpha$.
Its pure-input identity and the resulting bounds for permutation-invariant
inputs are proved in Subsection~\ref{subsec:symmetric-minimum-output}.
An infimum over strictly positive comparison states includes singular
limits; the support convention and regularization are specified below.
For fixed output weight, write
\begin{align}
 W_\alpha(\cM,\sigma)
 :=\max_{\rho\in\cD(A')}\omega_\alpha(\rho,\sigma,\cM).
 \label{eq:optimized-omega}
\end{align}
The maximum exists by continuity in $\rho$. Directly from the Choi tensor
product and multiplicativity of the Schatten norm,
\begin{align}
 \omega_\alpha(\rho_1\ox\rho_2,\sigma_1\ox\sigma_2,
                         \cM_1\ox\cM_2)
 =\omega_\alpha(\rho_1,\sigma_1,\cM_1)
          \omega_\alpha(\rho_2,\sigma_2,\cM_2).
 \label{eq:omega-product-inputs}
\end{align}
It remains to control the maximum over correlated inputs. We first use
covariance to choose an invariant optimizer, then apply the postselection
and finite-mixture argument of Section~\ref{sec:coherent-bound}.

\subsection{Covariance and invariant optimizers}

\begin{definition}\label{def: covariant}
Let $G$ be a compact group, and
for every $g\in G$, let $g\rightarrow U_{A}(g)  $ and
$g\rightarrow V_{B}(g)  $ be unitary representations of $G$ in $\mathcal{H}_{A}$ and $\mathcal{H}_{B}$,
respectively. Then a CPTP map $\mathcal{N}: {A\rightarrow B}$ is said to be
covariant with respect to these representations if the following relation
holds for all $g\in G$ and input density operators $\rho$:%
\begin{equation}
\mathcal{N} \!\left(  U_{A}(g)  \rho U_{A}^{\dag
}(g)  \right)  =V_{B}(g)  \mathcal{N}(  \rho)  V_{B}^{\dag}(g).
\end{equation}
\end{definition}

\begin{lemma}\label{lemma: optimal state of omega}
Let $1<\alpha<\infty$, let $\rho_{A'}$ be a density operator,
and assume $\sigma_B>0$; $\sigma_B$ need not have unit trace.
Let $\cM: A \to B$ be a covariant channel as in Definition~\ref{def: covariant}, and let $\sigma_B$ be invariant under the group action, i.e., $\sigma_B=V_B(g)^{\dagger} \sigma_B V_B(g)$ for all $g\in G$. Then,
\begin{align}
 \omega_{\alpha}\left(U_A^*(g)\rho_{A'} U_A^T(g), \sigma_B ,\cM \right)=\omega_{\alpha}(\rho_{A'},\sigma_B,\cM). \nonumber
\end{align}
By concavity of the function with respect to $\rho_{A'}$, we obtain
\begin{align}
 \omega_{\alpha}(\rho_{A'},\sigma_B,\cM) \leq \omega_{\alpha}\left(\sum_g p_g U_A^*(g)\rho_{A'} U_A(g)^{T},\sigma_B,\cM \right). \nonumber
\end{align}
Here $\{p_g\}$ is any finitely supported probability distribution
on $G$. The same conclusions hold for a completely positive map satisfying
the covariance identity, without requiring trace preservation.
%For a compact group, normalized Haar averaging gives
%\begin{align}
 %\bar\rho=\int_G U_A^*(g)\rho U_A^T(g)\,dg,
% \qquad \omega_\alpha(\rho,\sigma,\cM)
 %      \leq\omega_\alpha(\bar\rho,\sigma,\cM).
% \label{eq:omega-haar-optimizer}
%\end{align}
%The state $\bar\rho$ is invariant under $U_A^*$, so an invariant optimizer can be chosen.
\end{lemma}

\begin{proof}
The Schatten norm is invariant under unitary conjugation, i.e., $\norm{UXU^{\dagger}}_p=\norm{X}_p$. We apply this in the second line below
\begin{align}
 \omega_{\alpha}\left(U_A^* (g) \rho_{A'} U_A^{T}(g), \sigma_B,\cM\right)
 &=\norm{ \left( U_A^* (g)\rho_{A'}^{\frac{1}{2\alpha}}U_A^T (g)  \ox  \sigma_B^{\frac{1-\alpha}{2\alpha}} \right)  \ChM   \left( U_A^* (g)\rho_{A'}^{\frac{1}{2\alpha}}U_A^T (g) \ox  \sigma_B^{\frac{1-\alpha}{2\alpha}}\right)}_{\alpha}  \\
 &=\norm{ \left( \rho_{A'}^{\frac{1}{2\alpha}}U_A^T (g)  \ox  \sigma_B^{\frac{1-\alpha}{2\alpha}} \right)  \ChM   \left( U_A^* (g)\rho_{A'}^{\frac{1}{2\alpha}}\ox  \sigma_B^{\frac{1-\alpha}{2\alpha}}\right)}_{\alpha}  \\
 &=\norm{ \left( \rho_{A'}^{\frac{1}{2\alpha}}  \ox  \sigma_B^{\frac{1-\alpha}{2\alpha}}  V_{B}(g)\right)  \ChM   \left( \rho_{A'}^{\frac{1}{2\alpha}}\ox  V_{B}(g)^{\dagger}\sigma_B^{\frac{1-\alpha}{2\alpha}}\right)}_{\alpha}  \\
 &=\norm{ \left( \rho_{A'}^{\frac{1}{2\alpha}}  \ox V_{B}(g)^{\dagger} \sigma_B^{\frac{1-\alpha}{2\alpha}}  V_{B}(g)\right)  \ChM   \left( \rho_{A'}^{\frac{1}{2\alpha}}\ox  V_{B}(g)^{\dagger}\sigma_B^{\frac{1-\alpha}{2\alpha}} V_{B}(g) \right)}_{\alpha}  \\
 &=\omega_{\alpha}\left( \rho_{A'} , \sigma_B,\cM\right),
 %&=\omega_{\alpha}\left( \rho_{A'} ,V_B(g)^{\dagger} \sigma_B V_B(g),\cM\right)
 %&=\norm{ (\ChM)^{\frac{1}{2}}   \left( \rho_{A'}^{\frac{1}{\alpha}}  \ox   \sigma_B^{\frac{1-\alpha}{2\alpha}}\right)  (\ChM)^{\frac{1}{2}}    }_{\alpha}\\
\end{align}
where in the third line we apply the covariance property of the channel as well as the following property for a maximally entangled state: $(X\ox \1_A)\ket{\Phi}^{A'A}=( \1_{A'} \ox X^T)\ket{\Phi}^{A'A}$.
In the fourth line we again apply the invariance of the norm under unitary conjugation. 
The last line is due to the assumption $\sigma_B=V_B(g)^{\dagger} \sigma_B V_B(g)$ for all $g\in G$.

The  concavity  follows from Epstein's theorem
\cite{Epstein1973}. Put $r=1/\alpha$ and
$K=(\1_{A'}\ox\sigma^{(r-1)/2})J_\cM^{1/2}$. Then
\[
 \omega_\alpha(\rho,\sigma,\cM)^\alpha
 =\Tr\bigl[K^\dagger(\rho^r\ox\1_B)K\bigr]^{1/r}
\]
is concave in $\rho$. Taking the increasing concave power $r$ proves
concavity of $\omega_\alpha$. This uses only $J_\cM\geq0$, establishing
the completely positive extension.

By the above equality and that $\omega_{\alpha}$ is a concave function of  $\rho_{A'}$. 
We obtain
\begin{align}
\sum_g p_g \; \omega_{\alpha}(\rho_{A'},\sigma_B,\cM) &=\sum_g p_g \; \omega_{\alpha}\left(U_A^* (g) \rho_{A'} U_A^{T}(g), \sigma_B,\cM\right) \\
 & \leq \omega_{\alpha}\left(\sum_g p_g U_A(g)^*\rho_{A'} U_A(g)^{T},\sigma_B,\cM \right).
\end{align}

We note that $g\rightarrow U_{A}^*(g) $ is also a representation of group $G$.
%The same concavity inequality holds for the Haar integral.
%Left invariance of Haar measure makes $\bar\rho$ invariant under this
%representation, proving the optimizer assertion.

\end{proof}

\begin{corollary}\label{lemma: permutation invariant optimizer}
Let $\cM:A\to B$ be completely positive and let $\Sigma_n>0$ on $B^n$
be permutation invariant. For $\alpha>1$, a maximizing input in
\[
 \max_{\rho_n}\>\omega_\alpha(\rho_n,\Sigma_n,\cM^{\ox n})
\]
can be chosen permutation invariant. In particular this applies to
$\Sigma_n=\sigma^{\ox n}$ for every $\sigma>0$, including $\sigma=\1_B$.
\end{corollary}
\begin{proof}
Every tensor-power map is covariant under simultaneous permutations:
\[
 \cM^{\ox n}(P_A(\pi)X P_A(\pi)^\dagger)
 =P_B(\pi)\cM^{\ox n}(X)P_B(\pi)^\dagger,
 \qquad \pi\in S_n.
\]
The permutation matrices are real in the chosen tensor-product basis.
Lemma~\ref{lemma: optimal state of omega}, applied to the uniform average
$\bar\rho_n=(n!)^{-1}\sum_{\pi\in S_n}P_{A'}(\pi)\rho_nP_{A'}(\pi)^\dagger$,
therefore gives the assertion. No covariance assumption on the single-use
map is needed.
\end{proof}

\subsection{Tensor-power reduction by postselection}

\begin{proposition}
\label{prop:omega-polynomial}
Let $\cM:A\to B$ be completely positive, $d=|A|$, $\alpha>1$, and $n\geq1$.
Write $g_{n,d}=\binom{n+d^2-1}{d^2-1}$ for the postselection factor from
Section~\ref{sec:coherent-bound}.
If $\Sigma_n>0$ on $B^n$ is permutation invariant, then
\begin{align}
 \max_{\rho_n\in\cD({A'}^n)}
       \omega_\alpha(\rho_n,\Sigma_n,\cM^{\ox n})
 &\leq g_{n,d}
       \max_{\tau\in\cD(A')}
       \omega_\alpha(\tau^{\ox n},\Sigma_n,\cM^{\ox n}),
 \label{eq:omega-tensor-power-reduction}
\end{align}
In particular, for every fixed $\sigma>0$,
\begin{align}
 W_\alpha(\cM^{\ox n},\sigma^{\ox n})
 &\leq g_{n,d}
       \max_{\tau\in\cD(A')}
       \omega_\alpha(\tau^{\ox n},\sigma^{\ox n},\cM^{\ox n})
 \nonumber\\
 &=g_{n,d}W_\alpha(\cM,\sigma)^n.
 \label{eq:omega-polynomial}
\end{align}
For two completely positive maps $\cM_i:A_i\to B_i$ and weights $\sigma_i>0$,
\begin{align}
 &W_\alpha(\cM_1^{\ox n}\ox\cM_2^{\ox n},
                 \sigma_1^{\ox n}\ox\sigma_2^{\ox n})\nonumber\\
 &\qquad\leq g_{n,|A_1|}g_{n,|A_2|}
       W_\alpha(\cM_1,\sigma_1)^nW_\alpha(\cM_2,\sigma_2)^n.
 \label{eq:omega-polynomial-two}
\end{align}
\end{proposition}
\begin{proof}
By Corollary~\ref{lemma: permutation invariant optimizer}, choose a
permutation-invariant maximizing $\rho_n$. Set $r=1/\alpha\in(0,1)$ and
$g=g_{n,d}$, and let $T=T_{n,d}>0$ be the universal postselection marginal.
Lemma~\ref{lem:postselection-marginal} in
Appendix~\ref{app:supporting-inequalities} gives $\rho_n\leq gT$ and a
decomposition $T=\sum_iq_i\tau_i^{\ox n}$ with $L\leq g$, in which
$T$ commutes with every $\tau_i^{\ox n}$.
The operator monotonicity of $x^r$ therefore implies
$\rho_n^r\leq g^rT^r$.
To estimate $T^r$, put $B_i=q_i\tau_i^{\ox n}$.
Since $0\leq B_i\leq T$ and $[B_i,T]=0$, scalar functional calculus
in their common eigenbasis gives $T^{r-1}B_i\leq B_i^r$.
This uses commutation with $T$ separately for each term; the terms need
not commute with one another. Summing gives
\begin{align}
 \rho_n^r
 \leq g^rT^r
 =g^r\sum_iT^{r-1}B_i
 \leq g^r\sum_iq_i^r(\tau_i^r)^{\ox n}.
 \label{eq:postselection-power-step}
\end{align}
Tensor this order inequality with the positive operator $\Sigma_n^{r-1}$,
take the congruence by $J_{\cM^{\ox n}}^{1/2}$, and use monotonicity of
the Schatten norm on positive operators and its triangle inequality.
Equation~\eqref{eq:common-omega} then yields
\begin{align}
 \omega_\alpha(\rho_n,\Sigma_n,\cM^{\ox n})
 &\leq g^r\sum_iq_i^r
       \omega_\alpha(\tau_i^{\ox n},\Sigma_n,\cM^{\ox n})\nonumber\\
 &\leq g^rL^{1-r}\max_{\tau\in\cD(A')}
       \omega_\alpha(\tau^{\ox n},\Sigma_n,\cM^{\ox n})\nonumber\\
 &\leq g\max_{\tau\in\cD(A')}
       \omega_\alpha(\tau^{\ox n},\Sigma_n,\cM^{\ox n}),
\end{align}
where $\sum_iq_i^r\leq L^{1-r}$ follows from scalar concavity.
This proves Eq.~\eqref{eq:omega-tensor-power-reduction} directly by
postselection and Carath\'{e}odory, as in Section~\ref{sec:coherent-bound}.
For $\Sigma_n=\sigma^{\ox n}$, Eq.~\eqref{eq:omega-product-inputs}
factors each tensor-power value as $\omega_\alpha(\tau,\sigma,\cM)^n$,
proving Eq.~\eqref{eq:omega-polynomial}.

For two maps, average a maximizing input over $S_n\times S_n$, permuting
each channel's input copies and corresponding output copies separately.
The weight $\sigma_1^{\ox n}\ox\sigma_2^{\ox n}$ is invariant, so
Lemma~\ref{lemma: optimal state of omega} applies.
To make the joint postselection step explicit, let $R_i\cong A'_i$,
$P_i=P_{\Sym^n(R_iA'_i)}$, and
$|\Gamma_i\rangle=\sum_{a=1}^{|A_i|}|a\rangle^{R_i}\ox|a\rangle^{A'_i}$.
With tensor factors regrouped, the normalized canonical purification is
\[
 |\Psi\rangle
 =\bigl(\1_{R_1^nR_2^n}\ox\sqrt{\rho_n}\bigr)
   \bigl(|\Gamma_1\rangle^{\ox n}\ox|\Gamma_2\rangle^{\ox n}\bigr).
\]
For each $i$, simultaneous permutations of the $R_i$ and $A'_i$ copies
fix the maximally entangled vector. They also commute with
$\1_{R_1^nR_2^n}\ox\sqrt{\rho_n}$, since $\rho_n$ is invariant under
each input permutation action separately. Thus $|\Psi\rangle$ belongs to
$\Sym^n(R_1A'_1)\ox\Sym^n(R_2A'_2)$, and
\[
 (P_1\ox P_2)|\Psi\rangle=|\Psi\rangle,
 \qquad |\Psi\rangle\langle\Psi|\leq P_1\ox P_2.
\]
The purification may be entangled between the two symmetric subspaces;
the projector inequality uses only its support and unit norm.
Tracing out $R_1^nR_2^n$ and using
$\Tr_{R_i^n}P_i=g_{n,|A_i|}T_{n,|A_i|}$ gives
\[
 \rho_n\leq g_{n,|A_1|}g_{n,|A_2|}
              T_{n,|A_1|}\ox T_{n,|A_2|}.
\]
Take the product of the two finite ensembles in
Eq.~\eqref{eq:postselection-marginal-ensemble}. It contains at most
$g_{n,|A_1|}g_{n,|A_2|}$ terms, each commuting with
$T_{n,|A_1|}\ox T_{n,|A_2|}$. Repeating the preceding power and norm
estimates yields Eq.~\eqref{eq:omega-polynomial-two}.
\end{proof}

The prefactor $g_{n,d}$ is the same as in Section~\ref{sec:coherent-bound}
and is polynomial in $n$ for fixed $d$.
Thus covariance and postselection reduce the optimization
from all $n$-copy input states to tensor powers of a single input state.
The next argument removes this prefactor from the optimized product value.

\begin{theorem} 
\label{thm:omega-multiplicativity}
For finite-dimensional completely positive maps $\cM_i:A_i\to B_i$,
positive definite weights $\sigma_i$, and $\alpha>1$,
\begin{align}
 W_\alpha(\cM_1\ox\cM_2,\sigma_1\ox\sigma_2)
 =W_\alpha(\cM_1,\sigma_1)W_\alpha(\cM_2,\sigma_2).
 \label{eq:omega-multiplicativity}
\end{align}
\end{theorem}
\begin{proof}
Product maximizing inputs and Eq.~\eqref{eq:omega-product-inputs} give
supermultiplicativity. They also give
\begin{align}
 &W_\alpha(\cM_1\ox\cM_2,\sigma_1\ox\sigma_2)^n\nonumber\\
 &\quad\leq W_\alpha((\cM_1\ox\cM_2)^{\ox n},
                         (\sigma_1\ox\sigma_2)^{\ox n})\nonumber\\
 &\quad\leq g_{n,|A_1|}g_{n,|A_2|}
             [W_\alpha(\cM_1,\sigma_1)W_\alpha(\cM_2,\sigma_2)]^n,
\end{align}
where the last line is Eq.~\eqref{eq:omega-polynomial-two} after regrouping
input and output factors. Take $n$th roots and let $n\to\infty$.
The $n$th root of the polynomial tends to one, yielding the reverse
inequality. If either map is zero, both sides of
Eq.~\eqref{eq:omega-multiplicativity} vanish directly.
\end{proof}

\subsection{Completely bounded norms}

The first application is the completely bounded norm of a completely
positive map $\cM:A\to B$. Set $\sigma=\1_B$ in the common functional:
\begin{align}
 %a_\alpha(\cM):=
 \|\cM\|_{\CB,1\to\alpha}
 =\max_{\rho\in\cD(A')}\omega_\alpha(\rho,\1_B,\cM)
 =W_\alpha(\cM,\1_B).
 \label{eq:cb-variational}
\end{align}
Thus Theorem~\ref{thm:omega-multiplicativity} proves the first of the
three desired conclusions.

\begin{theorem}[Multiplicativity]\label{thm:cb-multiplicativity}
For finite-dimensional completely positive maps $\cM_1,\cM_2$ and
$\alpha>1$,
\begin{align}
 \norm{\cM_1\ox\cM_2}_{\CB,1\to\alpha}
 =\norm{\cM_1}_{\CB,1\to\alpha}
  \norm{\cM_2}_{\CB,1\to\alpha}.
 \label{eq:cb-multiplicativity}
\end{align}
\end{theorem}

\begin{proof}
Apply Theorem~\ref{thm:omega-multiplicativity} with
$\sigma_i=\1_{B_i}$ and use Eq.~\eqref{eq:cb-variational}.
Thus the CB norm is multiplicative; its logarithm is additive.
\end{proof}

Theorem~\ref{thm:cb-multiplicativity} recovers the multiplicativity theorem
of Devetak, Junge, King, and Ruskai~\cite{DevetakJungeKingRuskai2006}.
Here exact multiplicativity follows from permutation averaging, polynomial
domination by tensor powers, and an amplification argument. This is the
same reduction principle used in Section~\ref{sec:coherent-bound},
with the commutation property of
Lemma~\ref{lem:postselection-marginal} supplying the power estimate in
Eq.~\eqref{eq:postselection-power-step}.

\subsection{Channel divergence against a replacer}
\label{sec:replacer}

The second application fixes a density operator $\sigma_B$ as
the output weight. The resulting channel divergence is
$\alpha/(\alpha-1)$ times the logarithm of $W_\alpha(\cN,\sigma)$,
so multiplicativity of $W_\alpha$ gives its additivity.

Let $\cR_{\sigma}(X)=\Tr(X)\sigma_B$ and assume initially that
$\sigma_B>0$. Define the completely positive, generally non-trace-preserving map
\[
 \cM_{\alpha,\sigma}(X)
 =\sigma_B^{\frac{1-\alpha}{2\alpha}}\cN(X)
       \sigma_B^{\frac{1-\alpha}{2\alpha}}.
\]
The common Choi representation gives
\begin{align}
 \widetilde D_\alpha(\cN\mids\cR_\sigma)
 &:=\sup_{\psi^{RA}}\widetilde D_\alpha
       ((\id_R\ox\cN)(\psi^{RA})\mids\psi^R\ox\sigma_B)\nonumber\\
 &=\frac{\alpha}{\alpha-1}\log W_\alpha(\cN,\sigma)
 \nonumber\\
 &=\frac{\alpha}{\alpha-1}\log
          \|\cM_{\alpha,\sigma}\|_{\CB,1\to\alpha}.
 \label{eq:replacer-cb}
\end{align}
Here the supremum can be restricted to pure states with $R\cong A$.
This identity is established in Ref.~\cite{CooneyMosonyiWilde2016}.
It identifies the fixed-weight specialization of $\omega_\alpha$
with the tilted CB norm displayed above.
In particular, for channels $\cN_1,\cN_2$ and fixed states $\sigma_1,\sigma_2$,
\begin{align}
 \widetilde D_\alpha(\cN_1\ox\cN_2\mids
                  \cR_{\sigma_1}\ox\cR_{\sigma_2})
 =\widetilde D_\alpha(\cN_1\mids\cR_{\sigma_1})
  +\widetilde D_\alpha(\cN_2\mids\cR_{\sigma_2}).
 \label{eq:replacer-additivity}
\end{align}
Indeed, apply Theorem~\ref{thm:omega-multiplicativity}
to $(\cN_i,\sigma_i)$ and multiply its logarithm by
$\alpha/(\alpha-1)$. This proves the equality directly from the common
functional. Equivalently, the tilted map of the product is
$\cM_{\alpha,\sigma_1}\ox\cM_{\alpha,\sigma_2}$, so CB-norm
multiplicativity gives the same equality.
For singular $\sigma$, the same statement holds on its support if every output
of $\cN$ is supported there; otherwise the divergence is infinite.

\subsection{ R\'enyi channel mutual information}
\label{sec:renyi-mutual-information}
The third application minimizes over the output density operator.
We give an alternative postselection-based proof of channel R\'enyi
mutual-information additivity, preserving the original max--inf
optimization order. The additivity result was established by Gupta and
Wilde~\cite[Lemma~6]{GuptaWilde2015}.
For a channel $\cN$ and $1<\alpha<\infty$, define
\begin{align}
 b_\alpha(\cN)
 &:=\max_{\rho\in\cD(A')}\inf_{\substack{\sigma>0\\\Tr\sigma=1}}
                 \omega_\alpha(\rho,\sigma,\cN)
  =2^{\frac{\alpha-1}{\alpha}\widetilde I_\alpha(\cN)}.
 \label{eq:omega-mi-value}
\end{align}
The equality follows directly from Eq.~\eqref{eq:omega-mi-identification}.
All comparison states below are positive definite and have unit trace;
their infima include singular limits. Indeed, for
$\sigma_\delta=(1-\delta)\sigma+\delta\1_B/|B|$, $0<\delta<1$, and
$t=(\alpha-1)/\alpha$, operator antitonicity of $x^{-t}$ gives
\begin{align}
 \omega_\alpha(\rho,\sigma_\delta,\cN)
 \leq(1-\delta)^{-t}\omega_\alpha(\rho,\sigma,\cN).
 \label{eq:omega-regularization}
\end{align}
For singular $\sigma$, the right side is interpreted using the support
convention when finite, and is otherwise infinite. The factor tends to
one uniformly in $\rho$, justifying the use of full-rank approximations.

We first establish the finite-$n$ bound
\begin{align}
 b_\alpha(\cN^{\ox n})
 \leq g_{n,|A|}^{\,2-1/\alpha}b_\alpha(\cN)^n.
 \label{eq:mi-polynomial}
\end{align}
Put $r=1/\alpha$ and $g=g_{n,|A|}$. To justify the symmetry reduction
without interchanging optimizations, consider
\[
 f_n(\rho_n):=
 \inf_{\substack{\Sigma_n>0\\\Tr\Sigma_n=1}}
       \omega_\alpha(\rho_n,\Sigma_n,\cN^{\ox n}).
\]
For each fixed $\Sigma_n$, the functional is concave in $\rho_n$ by
the concavity argument in the proof of
Lemma~\ref{lemma: optimal state of omega}; its infimum $f_n$ is therefore
also concave. Channel covariance transforms the input and comparison
state by the corresponding permutations. Since the comparison-state
domain is unchanged by this transformation, $f_n$ is permutation
invariant. Averaging $\rho_n$ over permutations cannot decrease $f_n$,
so it suffices to bound permutation-invariant inputs.

For such an input, use the universal ensemble
$T=\sum_{i=1}^{L}q_i\tau_i^{\ox n}$, $L\leq g$, from
Lemma~\ref{lem:postselection-marginal}. For every $\varepsilon>0$,
choose positive definite density operators $\sigma_i$ such that
\[
 \omega_\alpha(\tau_i,\sigma_i,\cN)
 \leq \inf_\sigma\omega_\alpha(\tau_i,\sigma,\cN)+\varepsilon
 \leq b_\alpha(\cN)+\varepsilon.
\]
Each $\tau_i$ may have its own comparison state. Combine these choices
into the admissible permutation-invariant density operator
\begin{align}
 \Sigma_n:=\frac1L\sum_{i=1}^{L}\sigma_i^{\ox n}.
 \label{eq:mi-comparison-mixture}
\end{align}
For every $i$, $\Sigma_n\geq L^{-1}\sigma_i^{\ox n}$, and hence
\[
 \Sigma_n^{r-1}
 \leq L^{1-r}(\sigma_i^{r-1})^{\ox n}
\]
by operator antitonicity. Combining this estimate with
Eq.~\eqref{eq:postselection-power-step}, taking the Choi congruence,
and applying the Schatten triangle inequality and product-input identity
gives
\begin{align*}
 \omega_\alpha(\rho_n,\Sigma_n,\cN^{\ox n})
 &\leq g^rL^{1-r}\sum_{i=1}^{L}q_i^r
              \omega_\alpha(\tau_i,\sigma_i,\cN)^n\\
 &\leq g^rL^{2-2r}[b_\alpha(\cN)+\varepsilon]^n\\
 &\leq g^{2-r}[b_\alpha(\cN)+\varepsilon]^n,
\end{align*}
where $\sum_iq_i^r\leq L^{1-r}$. Taking the comparison-state infimum,
then the input maximum, and letting $\varepsilon\downarrow0$ proves
Eq.~\eqref{eq:mi-polynomial} in the original optimization order.

For two channels $\cN_j:A_j\to B_j$, average each input block
independently, as in Proposition~\ref{prop:omega-polynomial}.
The same concavity and covariance argument applies to the optimized
function $f_n$ for these two groups. Set $g_j=g_{n,|A_j|}$ and $G=g_1g_2$.
The joint input is dominated by $G T_1\ox T_2$, whose product ensemble
has $L_1L_2\leq G$ terms. Choose comparison states $\sigma_{j,a}$
separately for each single-copy ensemble state of channel $j$, with
functional value at most $b_\alpha(\cN_j)+\varepsilon$, and use
\[
 \Sigma_n=
 \left(\frac1{L_1}\sum_{a=1}^{L_1}\sigma_{1,a}^{\ox n}\right)
 \ox
 \left(\frac1{L_2}\sum_{b=1}^{L_2}\sigma_{2,b}^{\ox n}\right).
\]
Each product comparison term is dominated by $L_1L_2\Sigma_n$.
The preceding estimate therefore has prefactor
$G^r(L_1L_2)^{2-2r}\leq G^{2-r}$. Letting
$\varepsilon\downarrow0$ yields
\begin{align}
 b_\alpha((\cN_1\ox\cN_2)^{\ox n})
 \leq (g_1g_2)^{\,2-1/\alpha}
          [b_\alpha(\cN_1)b_\alpha(\cN_2)]^n,
 \label{eq:mi-polynomial-two}
\end{align}
after regrouping tensor factors. The following argument removes this
polynomial factor.

\begin{corollary}[Additivity of channel R\'enyi mutual information]
\label{cor:mutual-information-additivity}
For every $\alpha>1$,
\[
 \widetilde I_\alpha(\cN_1\ox\cN_2)
 =\widetilde I_\alpha(\cN_1)+\widetilde I_\alpha(\cN_2).
\]
\end{corollary}
\begin{proof}
For a bipartite state $\eta_{RB}$ write
$\widetilde I_\alpha(R;B)_\eta
 :=\inf_{\sigma_B}\widetilde D_\alpha(\eta_{RB}\mids\eta_R\ox\sigma_B)$.
State sandwiched R\'enyi mutual information satisfies
\begin{align}
 \widetilde I_\alpha(R_1R_2;B_1B_2)_{\eta_1\ox\eta_2}
 =\widetilde I_\alpha(R_1;B_1)_{\eta_1}
  +\widetilde I_\alpha(R_2;B_2)_{\eta_2}
\end{align}
by Ref.~\cite[Theorem~11]{Beigi2013}. Its minimization includes all states
on $B_1B_2$, including correlated ones. Choosing product pure channel
inputs therefore proves
$b_\alpha(\cN_1\ox\cN_2)\geq b_\alpha(\cN_1)b_\alpha(\cN_2)$.
It also gives $b_\alpha(\cP)^n\leq b_\alpha(\cP^{\ox n})$ for
$\cP=\cN_1\ox\cN_2$, using only the state result.
On the other hand, Eq.~\eqref{eq:mi-polynomial-two} gives
\begin{align*}
 b_\alpha(\cP)^n
 &\leq b_\alpha(\cP^{\ox n})\\
 &\leq \bigl(g_{n,|A_1|}g_{n,|A_2|}\bigr)^{2-1/\alpha}
          [b_\alpha(\cN_1)b_\alpha(\cN_2)]^n.
\end{align*}
Taking $n$th roots and letting $n\to\infty$ proves the reverse inequality,
since the $n$th root of the polynomial factor tends to one.
Finally, Eq.~\eqref{eq:omega-mi-value} converts multiplicativity of
$b_\alpha$ into the asserted additivity of $\widetilde I_\alpha$.
\end{proof}

\subsection{Minimum output entropy and pure input states}
\label{subsec:symmetric-minimum-output}

The additional objective here is to relate minimum output entropy to
$\omega_\alpha$ and use the postselection method of
Section~\ref{sec:coherent-bound} to bound the minimum output entropy
when the $n$-copy input is permutation invariant.
Let $\cN:A\to B$ be a channel and $\alpha>1$. Define
\begin{align}
 \nu_\alpha(\cN)
 &:=\max_{\substack{\rho\in\cD(A')\\\operatorname{rank}\rho=1}}
            \omega_\alpha(\rho,\1_B,\cN)
 =\max_{\tau\in\cD(A)}\|\cN(\tau)\|_\alpha.
 \label{eq:omega-max-output-norm}
\end{align}
To prove the second equality, observe that a pure $\rho$ satisfies
$\rho^{1/(2\alpha)}=\rho$ and
\begin{align}
 (\rho\ox\1_B)J_\cN(\rho\ox\1_B)
 &=\rho\ox\cN(\rho^T),&
 \omega_\alpha(\rho,\1_B,\cN)&=\|\cN(\rho^T)\|_\alpha.
 \label{eq:omega-pure-output}
\end{align}
Every pure input on $A$ occurs as $\rho^T$. Convexity of the Schatten
norm implies that its maximum over all input states is attained at a
pure state, proving Eq.~\eqref{eq:omega-max-output-norm}.
Thus, writing
\begin{align}
 %S_\alpha(\eta)&:=\frac{1}{1-\alpha}\log\Tr\eta^\alpha,&
 h_\alpha(\cN):=\min_{\tau\in\cD(A)}S_\alpha(\cN(\tau)),
\end{align}
we obtain the precise connection with the common functional:
\begin{align}
 h_\alpha(\cN)
 =\frac{\alpha}{1-\alpha}\log\nu_\alpha(\cN)
 =\frac{\alpha}{1-\alpha}\log
    \max_{\substack{\rho\in\cD(A')\\\operatorname{rank}\rho=1}}
          \omega_\alpha(\rho,\1_B,\cN).
 \label{eq:minimum-output-omega}
\end{align}

The pure-state constraint distinguishes $\nu_\alpha$ from
$W_\alpha(\cN,\1_B)$, whose maximization includes mixed states.
For example, for the identity channel on dimension $d>1$,
\[
 \nu_\alpha(\id_d)=1,\qquad
 W_\alpha(\id_d,\1)=\max_\rho\Tr\rho^{1/\alpha}
                    =d^{1-1/\alpha}.
\]
Averaging in Lemma~\ref{lemma: optimal state of omega} need not preserve
purity. We therefore impose permutation invariance as an input
restriction for the entropy problem and apply postselection directly.
Define
\begin{align}
 h_{\alpha,n}^{\mathrm{sym}}(\cN)
 :=\min_{\substack{\tau_n\in\cD(A^n)\\
                 P_A(\pi)\tau_nP_A(\pi)^\dagger=\tau_n\ \forall\pi\in S_n}}
           S_\alpha(\cN^{\ox n}(\tau_n)).
 \label{eq:symmetric-minimum-output-definition}
\end{align}
This minimization allows mixed permutation-invariant inputs.
At $\alpha=1$, use the von Neumann entropy in these definitions.

\begin{proposition}
\label{prop:symmetric-minimum-output}
For $d=|A|$, $n\geq1$, and $\alpha>1$,
\begin{align}
 nh_\alpha(\cN)-\frac{\alpha}{\alpha-1}\log g_{n,d}
 \leq h_{\alpha,n}^{\mathrm{sym}}(\cN)
 \leq nh_\alpha(\cN),\qquad
 g_{n,d}=\binom{n+d^2-1}{d^2-1}.
 \label{eq:symmetric-minimum-output-finite}
\end{align}
Consequently, for every finite $\alpha\geq1$,
\begin{align}
 \lim_{n\to\infty}\frac1n h_{\alpha,n}^{\mathrm{sym}}(\cN)
 =h_\alpha(\cN).
 \label{eq:symmetric-minimum-output-limit}
\end{align}
\end{proposition}
\begin{proof}
For a permutation-invariant $\tau_n$, symmetric purification and
postselection in Section~\ref{sec:coherent-bound}, followed by tracing
out the reference, give
\[
 \tau_n\leq g_{n,d}\sum_jq_j\tau_j^{\ox n},
 \qquad \sum_jq_j=1.
\]
This is also the domination in
Lemma~\ref{lem:postselection-marginal}, with the input copy relabeled.
Complete positivity, monotonicity of the Schatten norm on positive
operators, and its triangle inequality imply
\begin{align}
 \|\cN^{\ox n}(\tau_n)\|_\alpha
 &\leq g_{n,d}\sum_jq_j\|\cN(\tau_j)^{\ox n}\|_\alpha\nonumber\\
 &\leq g_{n,d}\nu_\alpha(\cN)^n\nonumber\\
 &=g_{n,d}\left[
       \max_{\substack{\rho\in\cD(A')\\\operatorname{rank}\rho=1}}
           \omega_\alpha(\rho,\1_B,\cN)\right]^n.
 \label{eq:symmetric-output-norm-omega}
\end{align}
Take logarithms and multiply by $\alpha/(1-\alpha)<0$.
Equation~\eqref{eq:minimum-output-omega} then proves the lower entropy
bound in Eq.~\eqref{eq:symmetric-minimum-output-finite}.
A tensor power of a single-use minimizing pure input is permutation
invariant and attains output entropy $nh_\alpha(\cN)$, proving the
upper bound. Since $\log g_{n,d}=O(\log n)$, division by $n$ proves
Eq.~\eqref{eq:symmetric-minimum-output-limit} for fixed $\alpha>1$.

For $\alpha=1$, $S(\eta)\geq S_\beta(\eta)$ for every $\beta>1$ gives
\[
 \liminf_{n\to\infty}\frac1n h_{1,n}^{\mathrm{sym}}(\cN)
 \geq h_\beta(\cN).
\]
The convergence $S_\beta\to S$ as $\beta\downarrow1$ is uniform on the
fixed-dimensional output state space, by joint continuity on its compact
product with $\beta\in[1,2]$. Hence $h_\beta(\cN)\to h_1(\cN)$.
Tensor powers of a single-use minimizing input give the reverse
asymptotic inequality.
\end{proof}

The result bounds the entropy reduction due to correlations in a
permutation-invariant input by $\alpha\log g_{n,d}/(\alpha-1)$ for
$\alpha>1$, and shows that this reduction vanishes per channel use.
Unrestricted minimum output entropy is nonadditive in general, for
R\'enyi orders $\alpha>1$~\cite{HaydenWinter2008} and for the von Neumann
entropy~\cite{Hastings2009}. The symmetry restriction in
Eq.~\eqref{eq:symmetric-minimum-output-definition} is therefore part
of this additional objective. Its connection to $\omega_\alpha$ is the
pure-input formula in Eq.~\eqref{eq:minimum-output-omega}; its bound uses
the same postselection step as the preceding three applications.

\section{Discussion}
\label{sec:discussion}

We have shown that full permutation invariance constrains the
coherent-information advantage obtainable from correlations across
channel uses. The optimized coherent information per
use converges to $Q^{(1)}(\cN)$, with a correction of order
$O(\log n/\sqrt n)$. Thus finite-block improvements remain possible,
while the asymptotic coherent-information optimization reduces to a
single channel use. For unrestricted inputs, superadditivity and Fekete's
lemma identify the supremum of the normalized coherent information over
block lengths with its limit as $n\to\infty$, yielding the regularized
quantum capacity~\cite{Lloyd_capacity_97,Shor_direct_capacity2002,Devetak-capacity-2005}.
For permutation-invariant inputs, the supremum
$\sup_{n\geq1}q_n^{\mathrm{sym}}(\cN)$ is attained at some finite block
length, possibly $n=1$. Indeed, any strict improvement over
$Q^{(1)}(\cN)$ exceeds all sufficiently late values of this convergent
sequence; if no such improvement occurs, a single use attains the
supremum. For the pure-source entanglement-generation codes
with permutation-invariant channel inputs defined in
Section~\ref{sec:symmetric-coding}, the same estimates give a strong
converse at rates above $Q^{(1)}(\cN)$.

The functional $\omega_\alpha$ places the channel quantities studied here
in a common variational framework for $\alpha>1$. Maximizing over the input with the identity as output
weight gives the completely bounded $1\to\alpha$ norm. A fixed density
operator as output weight yields, after input maximization and a
logarithmic rescaling, the channel divergence against the corresponding
replacer. Minimizing over output density operators before maximizing over
the input gives the channel R\'enyi mutual information with the same
rescaling.

For these three quantities, concavity and the permutation covariance of
every tensor-power channel allow invariant optimizers to be selected
without restricting the original channel optimization. Postselection then bounds
the optimized value by a polynomial factor times values on tensor-power
inputs. The universal postselection marginal commutes with each of these
tensor-power inputs, supplying the operator estimate needed for
$\omega_\alpha$. Applying the bound to repeated copies and taking roots
removes the polynomial cost and yields exact multiplicativity or additivity.

The additivity statements themselves are established results. Completely
bounded norm multiplicativity was proved by Devetak, Junge, King, and
Ruskai \cite{DevetakJungeKingRuskai2006}, with another proof through
conjugate channels due to Jen\v{c}ov\'a \cite{Jencova2006}; the channel
R\'enyi results appear in Refs.~\cite{GuptaWilde2015,CooneyMosonyiWilde2016}.
Our contribution is an alternative common proof strategy based on the
postselection estimates developed for symmetric coherent information.
For channel R\'enyi mutual information, we preserve the original max-inf
optimization order by constructing mixtures of tensor-power comparison
states; additivity of state R\'enyi mutual information supplies the
reverse inequality \cite{Beigi2013}. The common framework makes explicit
how the input optimization, the output weight, and tensor-product symmetry
combine to produce the respective additivity results.

A further application relates the pure-input restriction of
$\omega_\alpha$ to minimum output R\'enyi entropy. Postselection bounds
how much correlations in a permutation-invariant input can lower the
output entropy below $n$ times the single-use minimum, with a logarithmic
deficit for each fixed $\alpha>1$. The optimized entropy per use therefore
converges to the single-use minimum, and we also obtain convergence for
the von Neumann entropy. This symmetry-restricted conclusion complements
the failure of unrestricted minimum output entropy additivity
\cite{HaydenWinter2008,Hastings2009}. Together, the applications identify
concrete conditions under which permutation symmetry and polynomial
postselection bounds suffice to control correlations and recover
single-use asymptotic quantities.

\bigskip

\emph{Acknowledgments.}
Z.B.K. acknowledges support from the Ada Lovelace Postdoctoral Fellowship at Perimeter Institute for Theoretical Physics.  
GS acknowledges NSERC-NSF alliance grant ALLRP-586858-2023 and NSERC Discovery grant RGPIN-2025-02094.
DL is supported under an NSERC Discovery grant RGPIN-2024-03823 and an 
NSERC Alliance Consortia Quantum grants (ALLRP
578455-22).

Research at Perimeter Institute is supported in part by the Government of Canada through the Department of Innovation, Science and Economic Development and by the Province of Ontario through the Ministry of Colleges and Universities.

\emph{AI disclosure.} The ideas in this paper were obtained by the authors before the widespread use of AI in the research community. We have used AI to assist with the writing and preparation of the manuscript.

\appendix
\section{Supporting inequalities}
\label{app:supporting-inequalities}

The following symmetric-purification lemma is stated in~\cite[Lemma~II.5]{half_de_Finetti_2007}.

\begin{lemma}\label{lemma: symmetric purification}
\label{lemma: purification on symm}
Every permutation-invariant state on $A^{\ox n}$ admits a purification
in $\Sym^n(RA)$ with $R\cong A$ \cite{half_de_Finetti_2007}.
\end{lemma}
\begin{proof}
Take the canonical purification
$(\1_{R^n}\ox\sqrt{\rho^{A^n}})|\Gamma\rangle^{\ox n}$.
Its norm is one. Permutation matrices are real in the product basis, so
$|\Gamma\rangle^{\ox n}$ is fixed by simultaneous permutations of $R^n$
and $A^n$. Since $\sqrt\rho$ commutes with these permutations, so is the
displayed purification.
\end{proof}

\begin{lemma}\label{lemma: post-selection}
For a Hilbert space $\mathcal H$ of dimension $D$ and Haar probability
measure $d\varphi$ on its pure states,
 the standard Haar-moment identity
\cite{half_de_Finetti_2007} gives
\[
 \int\varphi^{\ox n}\,d\varphi
 =\frac{P_{\Sym^n(\mathcal H)}}{\binom{n+D-1}{D-1}}.
\]
Hence every state $\rho$ supported on $\Sym^n(\mathcal H)$ satisfies
$\rho\leq\binom{n+D-1}{D-1}\int\varphi^{\ox n}d\varphi$.
\end{lemma}
\begin{proof}
The integral is supported on the symmetric subspace, commutes with
$U^{\ox n}$ for every unitary $U$, and has trace one. Schur's lemma
identifies it with the normalized symmetric projector. A density operator
supported there is at most that projector \cite{post-selection2009}.
\end{proof}

\begin{lemma} 
\label{lem:postselection-marginal}
Let $n\geq1$ and $R\cong A'\cong\mathbb C^d$, and define
\begin{align}
 g_{n,d}&=\binom{n+d^2-1}{d^2-1},&
 T_{n,d}&=\frac{1}{g_{n,d}}\Tr_{R^n}P_{\Sym^n(RA')}.
 \label{eq:postselection-marginal}
\end{align}
Every permutation-invariant density operator $\rho_n$ on ${A'}^n$
satisfies $\rho_n\leq g_{n,d}T_{n,d}$. Moreover,
\begin{align}
 T_{n,d}=\sum_{i=1}^{L}q_i\tau_i^{\ox n},\qquad
 L\leq g_{n,d},\qquad
 [T_{n,d},\tau^{\ox n}]=0\quad\text{for every density operator }\tau,
 \label{eq:postselection-marginal-ensemble}
\end{align}
where $\{q_i\}$ is a probability distribution and $\tau_i\in\cD(A')$.
\end{lemma}
\begin{proof}
Apply the symmetric purification and postselection lemmas used in
Section~\ref{sec:coherent-bound}
(Lemmas~\ref{lemma: symmetric purification} and~\ref{lemma: post-selection})
and trace out $R^n$. This gives $\rho_n\leq g_{n,d}T_{n,d}$ and
\[
 T_{n,d}=\int(\Tr_R\varphi_{RA'})^{\ox n}\,d\varphi,
\]
where the integral is over single-copy pure states with Haar probability
measure~\cite{post-selection2009}.
Apply Carath\'{e}odory's theorem to this marginal integral.
To count the required terms, choose a Hermitian basis of the $d^2$-dimensional
real space of Hermitian operators on $A'$. Symmetrized tensor products of
these basis elements are indexed by the multisets of size $n$ from $d^2$
elements. They form a basis of the permutation-invariant Hermitian
operators, so that space has dimension $g_{n,d}$. Its trace-one affine
subspace has dimension $g_{n,d}-1$, giving $L\leq g_{n,d}$.
This applies Carath\'{e}odory after tracing out the reference, which sharpens
the sufficient $g_{n,d}^2$ term count used in Section~\ref{sec:coherent-bound}.

Finally, fix $\tau\in\cD(A')$ and write
$X=\1_R\ox\tau$ and $P=P_{\Sym^n(RA')}$.
Applying the same operator to each $RA'$ factor preserves symmetry,
so $X^{\ox n}$ preserves $\Sym^n(RA')$. Since $X^{\ox n}$ is Hermitian,
it also preserves the orthogonal complement: for
$|s\rangle\in\Sym^n(RA')$ and
$|a\rangle\in(\Sym^n(RA'))^\perp$,
\[
 \langle s|X^{\ox n}|a\rangle
 =\langle X^{\ox n}s\mid a\rangle=0,
\]
because $X^{\ox n}|s\rangle\in\Sym^n(RA')$.
Thus $PX^{\ox n}=X^{\ox n}P$.
Regrouping the tensor factors gives
$X^{\ox n}=\1_{R^n}\ox\tau^{\ox n}$. Taking the partial trace over $R^n$
and pulling out the operators acting on ${A'}^n$, we obtain
\begin{align}
 [T_{n,d},\tau^{\ox n}]
 &=\frac{1}{g_{n,d}}\Tr_{R^n}
       \bigl[P,\1_{R^n}\ox\tau^{\ox n}\bigr]=0.
 \label{eq:postselection-commutation}
\end{align}
Also $T_{n,d}\geq\1_{{A'}^n}/(g_{n,d}d^n)>0$, since the tensor
power of a normalized maximally entangled vector belongs to
$\Sym^n(RA')$ and its projector is bounded by the symmetric projector.
\end{proof}

\begin{lemma}
\label{lemma: H_alpha for rho<gamma rho'}
\label{lemma: D < max_y D_y}
For states $\rho,\rho'$ with $\rho\leq c\rho'$ and $\alpha>1$,
\begin{align}
 \widetilde D_\alpha(\rho\mids\tau)
 &\leq\widetilde D_\alpha(\rho'\mids\tau)
              +\frac{\alpha}{\alpha-1}\log c.
 \label{eq:order-inequality}
\end{align}
For a probability distribution $q_j$, states $\rho_j$, and positive
operators $\tau_j$,
\begin{align}
 &\widetilde D_\alpha\left(\sum_jq_j\rho_j\ox|j\rangle\langle j|
            \mids\sum_jq_j\tau_j\ox|j\rangle\langle j|\right)\nonumber\\
 &\qquad=\frac{1}{\alpha-1}\log\sum_jq_j
       2^{(\alpha-1)\widetilde D_\alpha(\rho_j\mids\tau_j)}
 \leq\max_j\widetilde D_\alpha(\rho_j\mids\tau_j).
 \label{eq:flag-inequality}
\end{align}
\end{lemma}
\begin{proof}
For the first inequality, conjugate the order relation by
$\tau^{(1-\alpha)/(2\alpha)}$ and use monotonicity of the Schatten norm
on positive operators. For the second, evaluate the trace of the
$\alpha$th power block by block. The weights contribute
$q_j^\alpha q_j^{1-\alpha}=q_j$; their sum is one.
\end{proof}

\setstretch{1.0}
\bibliography{lib}
\end{document}